\documentclass[12pt]{article}
\usepackage{graphicx}
\usepackage{epsfig}
\usepackage{amssymb}
\usepackage{amsmath}
\usepackage{times}
\usepackage{./mhequ}
\usepackage{subfig}

\newenvironment{proof}[1][Proof]{\noindent\textit{#1.} }{\ \rule{0.5em}{0.5em}\par}

\newcommand{\ed}{{\rm e}}

\newtheorem{theorem}{Theorem}[section]
\newtheorem{lemma}[theorem]{Lemma}

\newtheorem{corollary}[theorem]{Corollary}

\newcommand{\myl}[1]{
\hspace{-1mm}
\left(\vbox to #1pt{}\right.
\hspace{-1mm}
}

\newcommand{\my}[2]{
\left#1\vbox to #2pt{}\right.
}

\newcommand{\myr}[1]{
\hspace{-1mm}
\left.\vbox to #1pt{}\right)
\hspace{-1mm}
}

\begin{document}
\title{The QED $\beta$-function from global solutions to Dyson-Schwinger equations}

\author{Guillaume van Baalen\footnote{Department of Mathematics and
    Statistics, Boston University, 111 Cummington Street, Boston MA,
    02215, USA}, Dirk Kreimer\footnote{CNRS-IHES 91440 Bures sur
    Yvette, France; and Center for Mathematical Physics Boston
    University, 111 Cummington Street, Boston MA,
    02215, USA}, David Uminsky\footnotemark[1], and Karen
  Yeats\footnotemark[1]}
\maketitle
\begin{abstract}
We discuss the structure of beta functions as determined by 
the recursive nature of Dyson--Schwinger equations turned into an 
analysis of ordinary differential equations, with particular emphasis 
given to quantum electrodynamics.
In particular we determine when a separatrix for  solutions to 
such ODEs exists and clarify the existence of
Landau poles beyond perturbation theory.  Both are determined in terms
of explicit conditions on the asymptotics 
for the growth of skeleton graphs.
\end{abstract}
\section*{Acknowledgments} G.vB. and D.U. would like to thank C. Eugene
Wayne for
insightful conversations.  D.K. and K.Y. would like to thank Christoph
Bergbauer, David Broadhurst, Ivan Todorov, and Spencer Bloch.  K.Y.
would also like to thank
Cameron Morland for for help visualizing the problem.  D.K. and
K.Y. are supported by NSF grant DMS-0603781.  D.U. is supported by NDF
grant DMS-0405724.
\section{Introduction}
\subsection{The method}
Results on the structure of amplitudes in the theory of local
interacting quantum fields are notoriously hard to come by beyond
perturbation theory. We refrain from discussing the various
approaches developed in the past and shortly summarize our
approach here, which has been developed by one of us (D.K.) in
the last decade
\cite{BDK,habil,hopf,ckII,bkerfc,bergk,etude,radii,tor}. 
It lead already to progress at very high orders \cite{bk30,bellonI, bellonII} 
and all orders of perturbation theory \cite{bkerfc} (see also the
$P=x$, $s=2$ 
case in the examples below).

It is a
pleasure to emphasize that our approach connects to old attempts
\cite{MT} in quantum field theory to use the soft breaking of
conformal symmetry by renormalizable quantum fields for
non-perturbative results. The recent developments which allow us
to understand the notion of locality mathematically combine
rather nicely with such ideas. A crucial ingredient is that the
mathematical structure of the quantum equations of motion remains
form-invariant under inclusion of more and more skeleton graphs,
and this fact allows the development of an approximation to these
equations in terms of periods of increasing complexity, without
ever changing the structure of these equations. This is very
different from, for example, any truncation of high frequency
modes in the path-integral. Whilst the approach used here can
re-derive results of such constructive methods \cite{radii}, we
here go beyond what is possible by such truncations of the
path-integral.

In particular, for theories which are non-asymptotically free, a 
study of low orders of perturbation theory indicates the presence of a 
Landau pole (the invariant charge approaching infinity at a finite scale 
$q^2/\mu^2$), which is also believed to exist for such theories in the 
constructive approach, if one attempts to remove the cut-off which 
necessarily has to be introduced in such theories, as well as in 
perturbation theory. In our approach, we only choose a boundary 
condition for the equation of motion, the Dyson--Schwinger equations. We 
approximate the full theory by the choice of a function $P(x)$ which 
describes the growth of the skeleton expansion, and carefully make that 
choice to maintain the Lie- and Hopf-algebraic structure of the forest 
formula and the equations of motion at the same time. We thus do not 
need to introduce a cut-off, a familiar phenomenon when scaling 
dimensions of Green functions are taken into account in those equations 
\cite{MT}.  Perturbative approximations to $P(x)$ lead to a 
non-perturbative behavior for $\beta$ functions in such theories which 
reconfirms the existence of a Landau pole. Rather mild assumptions on 
the non-perturbative behavior of $P(x)$ allow for solutions though which 
avoid such a pole, as discussed below, with the charge going to
infinity 
only at infinite scale, and hence realizing a possibility already 
discussed in \cite{Weinberg} section 18.3. Finally, we emphasize that we assume below 
that $P(x)$ is a nowhere vanishing function, and hence that we do not 
have a non-trivial zero for the $\beta$-function for a 
non-asymptotically free theory: so we are not assuming an eigenvalue 
condition \cite{Adler}, but much to the contrary, analyze the structure 
of the theory under the assumption that such an eigenvalue does not 
exist.

We will not attempt any serious discussion of the asymptotics of $P(x)$,
though that the work of \cite{ipz}
emphasizes the need of such a discussion. Here, we are content with the 
classifications of the behavior of the $\beta$-function as a function 
of  the possible asymptotics of $P(x)$,
emphasizing the possibility of the absence of Landau poles in 
well-specified conditions. Also, we re-emphasize that Dyson--Schwinger 
equations do not demand the introduction of a cut-off, but rather demand 
the specification of a finite number of conditions to fix the amplitudes 
needing renormalization as initial conditions for the renormalization 
group flow.

So our approach is based on two main ingredients: the existence
of quantum equations of motion --- Dyson--Schwinger equations
---, and the consequences of the renormalization group for such
local field theories. The latter guarantee that amplitudes
develop anomalous scaling exponents under the action of the
dilatation group which re-scales parameters in the theory, the
former guarantees sufficient recursive structure in the theory
such that a non-perturbative approach becomes feasible. The rich
Hopf algebraic foundations of these phenomena make our approach
possible.

We consider Green functions as functions of two variables, a
`running coupling constant' $x$ and a single kinematical variable
$L=\ln q^2/\mu^2$ (in the deep Euclidean regime, or suitably
continued to physical regions). This implies that vertex
functions are considered only at zero momentum transfer or for
symmetric external momenta.

We define Green functions as the scalar coefficient functions
which provide quantum corrections to tree-level amplitudes
$r\in{\mathcal R}$, where $r$ denotes the chosen amplitude. We
store all the information on parameters which determine the
amplitude under consideration in its lowest order contribution,
the tree level form-factor $f(r)$. The Lagrangian is then given
as ${\mathcal L}=\sum_{r\in\mathcal R}m{r}$, for monomials $m$ in
fields for each amplitude which needs renormalization.
Green functions modify this amplitude $f(r)$ in a multiplicative
manner: $f(r)\to \Phi(r)(1+{\mathcal O}(\hbar))$, and hence start
with one.

The equivalent expansions (taking the negative sign for
propagators and the positive for vertices)
\begin{equs}
G^r( x ,L)=1\pm\sum_{j=1}^\infty \gamma_j^r( x
)L^j=1\pm\sum_{j=1}^\infty c_j^r(L) x ^j~,
\end{equs}
where $c_j^r(L)$ is a polynomial in $L$ bounded in degree by $j$
and $\gamma_j^r$ a series in $ x $, are triangular and recursive
for $\gamma^r_j$: the renormalization group determines the
$\gamma_j^r$, $j>1$ in terms of all the series
$\gamma_1^r$. We denote $\gamma_1^r$ as the anomalous dimensions
of the amplitude $r$, even if $r$ is a vertex function.
Any $r$ for a vertex amplitude corresponds to a field monomial
$m(r)=\prod_i\eta_i$ in the Lagrangian, and the $\eta_i$ are
fields which we assume to have kinetic energy. There are then
corresponding monomials $\sim \eta_i^2$ quadratic in those fields
$\eta_i$ in the Lagrangian and the corresponding
$\gamma_1^i\equiv\gamma_1^{\eta_i^2}$ combine with $\gamma_1^r$
to give the $\beta$-function in our sign conventions as
\begin{equs}
\beta^r=x[\gamma_1^r+\sum_i \gamma_1^{i}/2]~.
\end{equs}
Here, the monomial $r$, for $r$ a vertex amplitude, comes along
with a probability $x$ for the tree-level scattering process
described by $r$ to happen, and this probability --- necessarily
smaller than one ---, furnishes a natural expansion parameter for
the series $\gamma_1$. We consider only theories which have a
single vertex amplitude in this paper, and hence have a unique
expansion parameter.

Following now \cite{etude}, \cite{radii}, and \cite{kythesis} we
can reduce the Dyson--Schwinger equations to a system of
differential equations for the anomalous dimensions $\gamma_1^r$.

We outline the argument as follows. As we are interested only in
the high-energy sector of the theory, we reduce one-particle
irreducible Green functions to depend on a single scale $L = \log
q^2/\mu^2$. The coupling constant will be denoted $x$. Then we
expand the Green functions in $L$ as above.
%a
%next-to next-to $\cdots$ leading log expansion with the notation above,
%\[
%  G^r(x, L) = 1 \pm \sum \gamma^r_k L^k
%\] where $r$ will represent any amplitude which needs renormalization and
%the sign is positive for a vertex function and negative for a
%propagators, which are the inverse of one-particle irreducible Green functions.

From the renormalization group equation we obtain \cite{etude,kythesis}
\begin{equs}\label{gamma k recursion}
  \gamma^r_k(x) =-\frac{1}{k}\left( \pm\gamma^r_1(x) + \sum_{j
      \in \mathcal{R}}|s_j|\gamma^j_1(x)x
      \partial_x\right)\gamma^r_{k-1}(x)~.
\end{equs}
where again the sign is positive for a vertex and negative for a
propagator and where the $s_j$ are defined, in accordance with
the fields coupling at a vertex as above, by
\begin{equs}
\beta(x) = x\sum_{j \in \mathcal{R}} |s_j| \gamma^j_1(x)
\end{equs}
where $\beta(x)$ is said $\beta$-function of the theory, and
${\mathcal R}$ is the set of all amplitudes needing
renormalization. In the single equation case \eqref{gamma k
recursion} reads
\begin{equs}
  \gamma_k = \pm\frac{1}{k} \gamma_1(x)(1-sx
  \partial_x)\gamma_{k-1}(x)~.
\end{equs}  
As in \cite{radii} the Dyson-Schwinger equations can be rewritten
in terms of derivatives of the Mellin transforms for the
primitives. By Mellin transforms we simply mean the analytically
regularized Feynman integrals for the primitives. Again by
combinatorics on the Hopf algebra we can reduce to Mellin
transforms in a single variable $\rho$, that is to a single
insertion place. Finally by shifting unwanted powers of $\rho$
at a given loop order to lower powers of $\rho$ at a higher loop
order, as in \cite{radii}, \cite{kythesis}, we can relate the
coefficients of $L$ and $L^2$ which, in view of \eqref{gamma k recursion} gives us the system
\begin{equs}\label{system}
  \gamma_1^r(x) = P_r(x) \mp \gamma_1^r(x)^2 + \left(\sum_{j \in \mathcal{R}}|s_j|\gamma_1^j(x)\right) x \partial_x \gamma_1^r(x)
\end{equs}
as $r$ runs over $\mathcal{R}$, the residues of the theory. $P_r$
is a modified version of the values of the primitives. The
modification comes from two places. First the reduction to a
single insertion place is purely combinatorial and leads
essentially to the need to consider in the contribution to $P_r$
primitives which are not merely single graphs, but sums of graph.
Second by exchanging powers of $\rho$ for powers of the coupling
constant $x$ we also modify $P_r$. This reduction is not yet as
well understood, but none-the-less it is simply a rearrangement
of the analytic information contained in the original primitives
\cite{kythesis}.

\subsection{QED as a special case}
Most of our analysis will be confined to the case with only one
equation and $s>0$.
\begin{equs}\label{de}
  \gamma_1(x) = P(x) -\gamma_1(x)^2 + s\gamma_1(x) x \partial_x \gamma_1(x)
\end{equs}
This case is general enough to cover gauge theories for the
following reasons. First of all, we note that in QED thanks to
the Ward identity, the $\beta$-functions is computable from the
anomalous dimension of the photon $\gamma_1^{\frac{1}{4}F^2}$
alone, and similarly for non-abelian gauge theories in a
background field gauge \cite{Abbott}. Furthermore, using a
Baker--Johnson--Willey gauge \cite{jbw}, QED is a finite theory
were it not for the (gauge-invariant) photon propagator. So it is
indeed a single equation case.

In particular, it is the case $s=1$, reflecting the fact that the
lowest order term in the $\beta$-function comes from a graph
which itself has no internal photon propagation. The variable $s$
measures the power of the Green function appearing in the
invariant charge, or the power of the nonlinear part of the
recursive appearance of the Green function in the Dyson-Schwinger
equation. More specifically the power of the recursive appearance
of the Green function at loop order $k$ is $1-sk$. So with $s=1$
and $k=1$ we the power is $0$ representing the fact that in QED
the one loop photon graph doesn't have an internal photon edge
and the correct counting continues to hold at higher loop orders.
For the Yukawa theory example of \cite{bkerfc} the Green function
appears recursively with power $-1$ at $k=1$ leading to $s=2$.

In the general non-abelian case the co-ideal and Hochschild
cohomology structure of the Hopf algebra underlying the expansion
of a non-abelian gauge theory in the coupling \cite{anatomy,vS,vS2}
allow for similar simplifications in particular in the background
field method. An analysis of this method will be given in future
work.

Notice that the $\beta$-function is showing up as the coefficient
of $(\gamma_1^r)'(x)$ in (\ref{system}) above, namely
\begin{equs}
  \beta(x) = x\sum_{j \in \mathcal{R}} |s_j| \gamma_1^j(x)
\end{equs}
in the system case and
\begin{equs}
  \beta(x) = xs\gamma_1(x)
\end{equs}
in the single equation case with $s>0$. Consequently this differential
equation is well suited to improving our understanding of the
$\beta$-function. Furthermore, solving for the derivative
${\gamma_1^r}^\prime(x)$ shows the appearance of the
$\beta$-function in the denominator. A zero of the
$\beta$-function is hence a degenerate singular case, even if it
corresponds to a simple scaling behavior of quantum fields
reflected by an abelian renormalization group flow (or,
equivalently, a co-commutative expansion in the Hopf algebra of
perturbation theory) \cite{linetude}.

In particular in the single equation case, we see immediately
from \eqref{de} that any zeroes of $\beta(x)$ must occur either
where $P(x)=0$ --- so we have no quantum corrections driving the
equations of motion at some particular value of the coupling ---
or where $\gamma_1'(x)$ is infinite. The second of these
possibilities is not physically reasonable for a finite value of
$x$ and is indeed only realized by solutions for $\gamma_1^r$
which are multi-valued.

For QED taken out to four loops and correcting the primitives for our
setup and using values from \cite{gkls} we have
\begin{equs}
  P(x) = \frac{x}{3} + \frac{x^2}{4} +
  (-0.0312 + 0.06037)x^3 +(-0.6755 + 0.05074)x^4~.
\end{equs}
$P(x)$ is
decreasing for $x>0.653\ldots$, which will not be permissible below, and it has a zero at
$x=0.992\ldots$.  We expect the zero to be spurious as it would
immediately lead to a zero of the $\beta$-function, and believe
both problems are due only to taking the 4 loop approximation
beyond where it is valid.  On the other hand the estimate 
\begin{equs}
c_k \sim_{k\>> 1}(-1)^kk!k^3~,
\end{equs}
where $P(x)=\sum_k c_k x^k$, of \cite{ipz} suggests via resummation
that $P(x)$ is
bounded for large $x$, which tantalizingly leads to the possibility of the 
absence of a Landau pole in an analysis of the $\beta$-function beyond 
perturbation theory.  The behavior $P(x)$ in the QED case, clearly deserves 
more attention from this viewpoint.

The system case is not quite so simple. Assume $\beta(x)=0$. If
we rule out infinite $(\gamma_1^r)'(x)$, then we can only
conclude that for each $r \in \mathcal{R}$
\begin{equs}
  \gamma_1^r(x) + \gamma_1^r(x)^2 - P_r(x) = 0~.
\end{equs}

\subsection{Exposition of the main results}
In the remainder of this article, we will restrict ourselves to
the single equation case \eqref{de} with $s>0$, which we can rewrite as
\begin{equs}
\label{preDSeqn}
\frac{{\rm d}\gamma_1}{{\rm d}x}=f(\gamma_1(x),x)~,
~~~\mbox{ where }~~~
f(\gamma_1,x)=
\frac{\gamma_1
+\gamma_1^2-P(x)}{sx\gamma_1}~.
\end{equs}
The main assumptions we make on the primitive skeleton function are:
\begin{itemize}
\item[H1:] $P$ is a twice differentiable function on ${\bf R}^{+}$, with 
$P(0)=0$ and $P(x)>0$ if $x>0$.
\item[H2:] $P$ is everywhere increasing.
\end{itemize}

To motivate the detailed study of (\ref{preDSeqn}) we will
conduct in the remainder of this paper, we consider briefly the
simple examples $P(x) = x$ and $s=1$ or $s=2$. For a qualitative
overview, see Figure \ref{alltogether}.

\begin{figure}
\begin{center}
\unitlength1mm
\begin{picture}(190,90)(0,0)
\put( 10,0){\epsfig{file=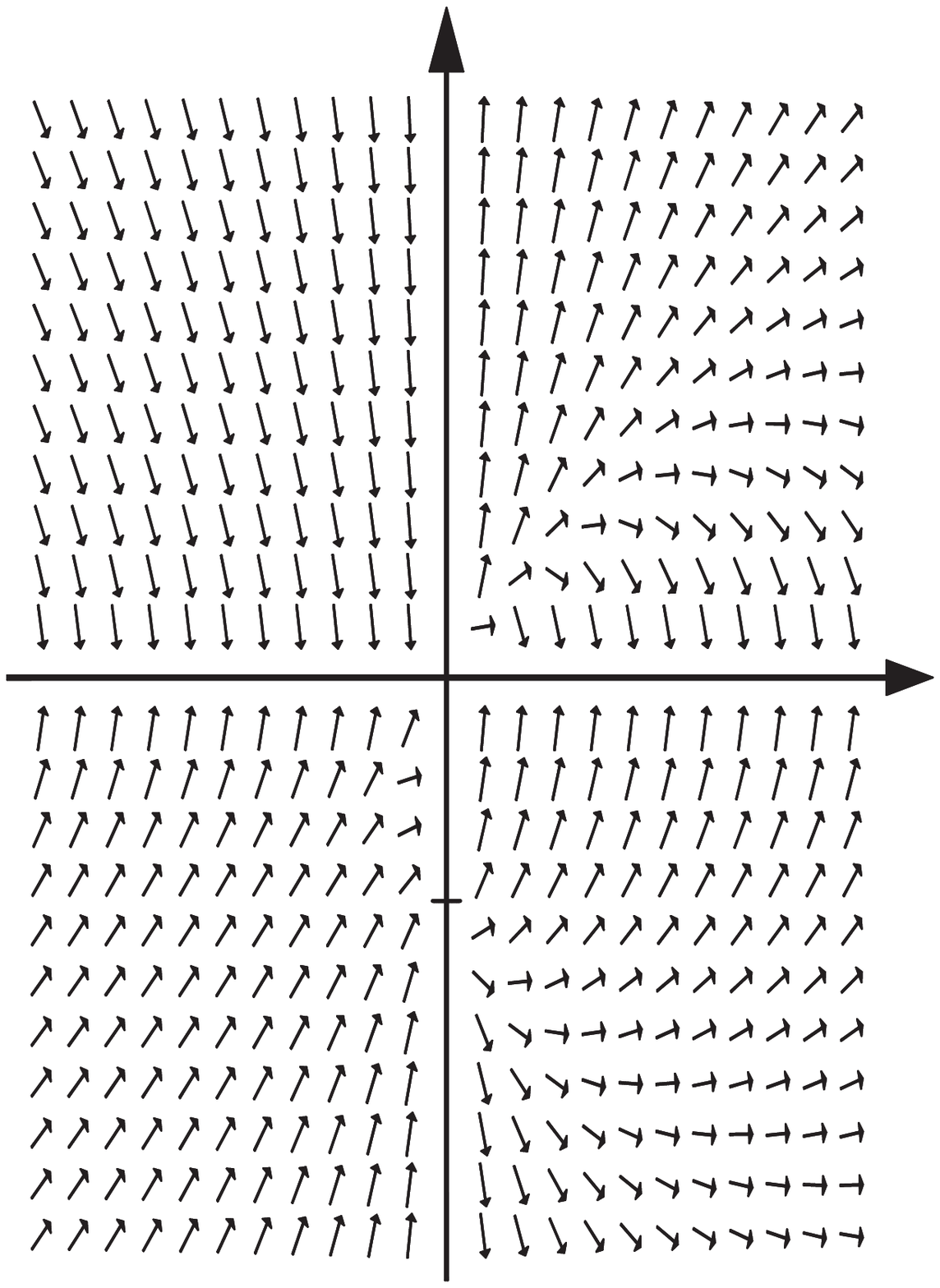,height=8.5cm}}
\put( 90,0){\epsfig{file=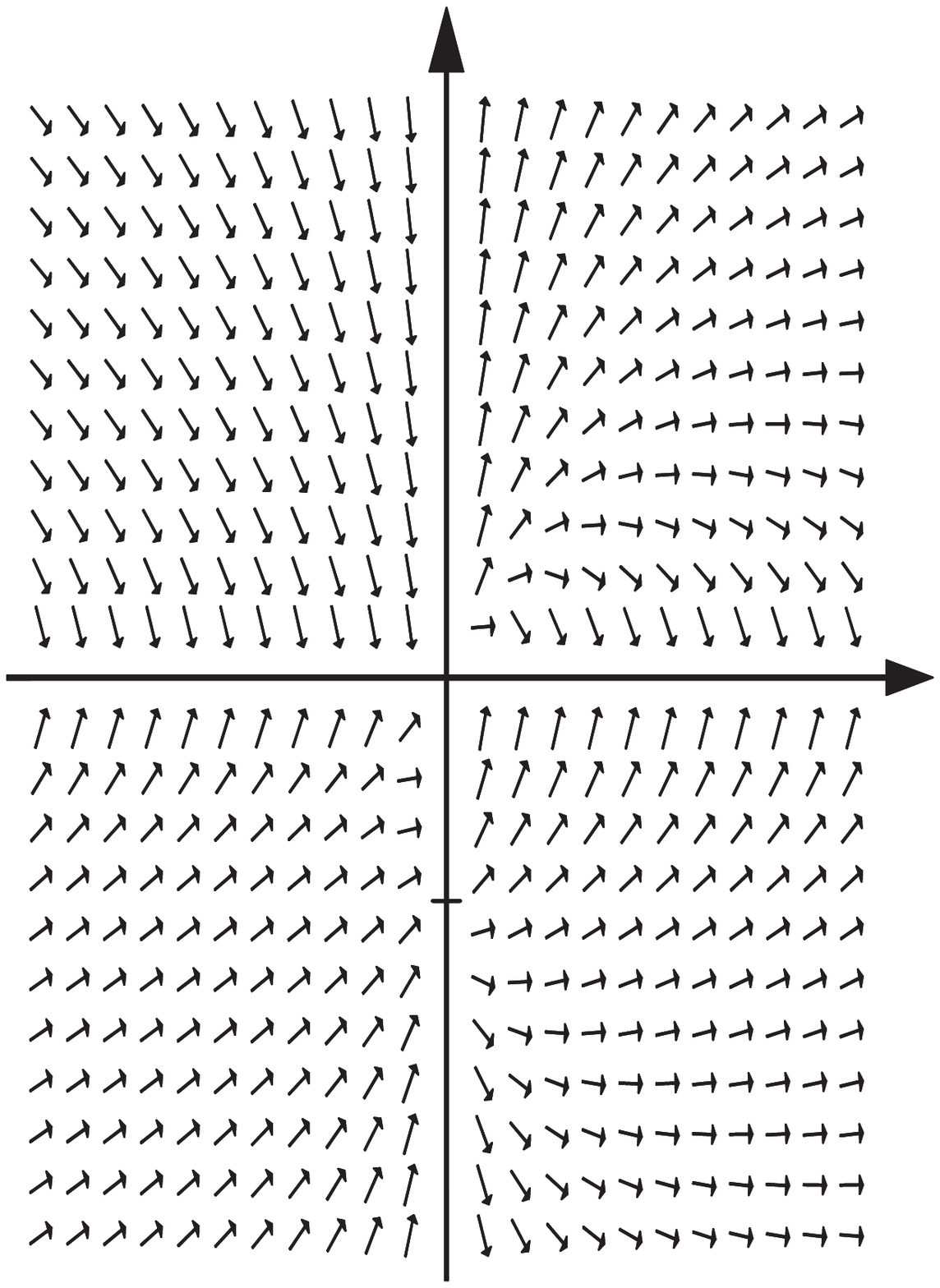,height=8.5cm}}
\put(69,43){$x$}
\put(149,43){$x$}
\put(42,82){$\gamma_1$}
\put(122,82){$\gamma_1$}
\end{picture}
\end{center}
\caption{Direction fields for $P(x)=x$, $s=1$ on left panel,
$s=2$ on right panel. Note how the two pictures look similar;
yet, as will be shown, the $s=1$ case admits global positive
solutions, and $s=2$ does not.}\label{alltogether}
\end{figure}

If $s=1$, we can solve (\ref{preDSeqn}) by specifying
$\gamma_1(x)$ at $x=1$, finding
\begin{equs}
\gamma_1(x)=x + x~W
\myl{12}(\gamma_1(1)-1)\exp\myl{10}\gamma_1(1)-{\textstyle\frac{1}{x}}\myr{10}
\myr{12}~,
\label{eqn:solsol}
\end{equs}
where $W$ is the Lambert $W$ function. A few of such solutions,
along with the direction field associated with (\ref{preDSeqn}) are
displayed in figure \ref{fig:solutionsnonsystem}.
\begin{figure}
\begin{center}
\unitlength1mm
\begin{picture}(190,90)(0,0)
\put( 10,0){\epsfig{file=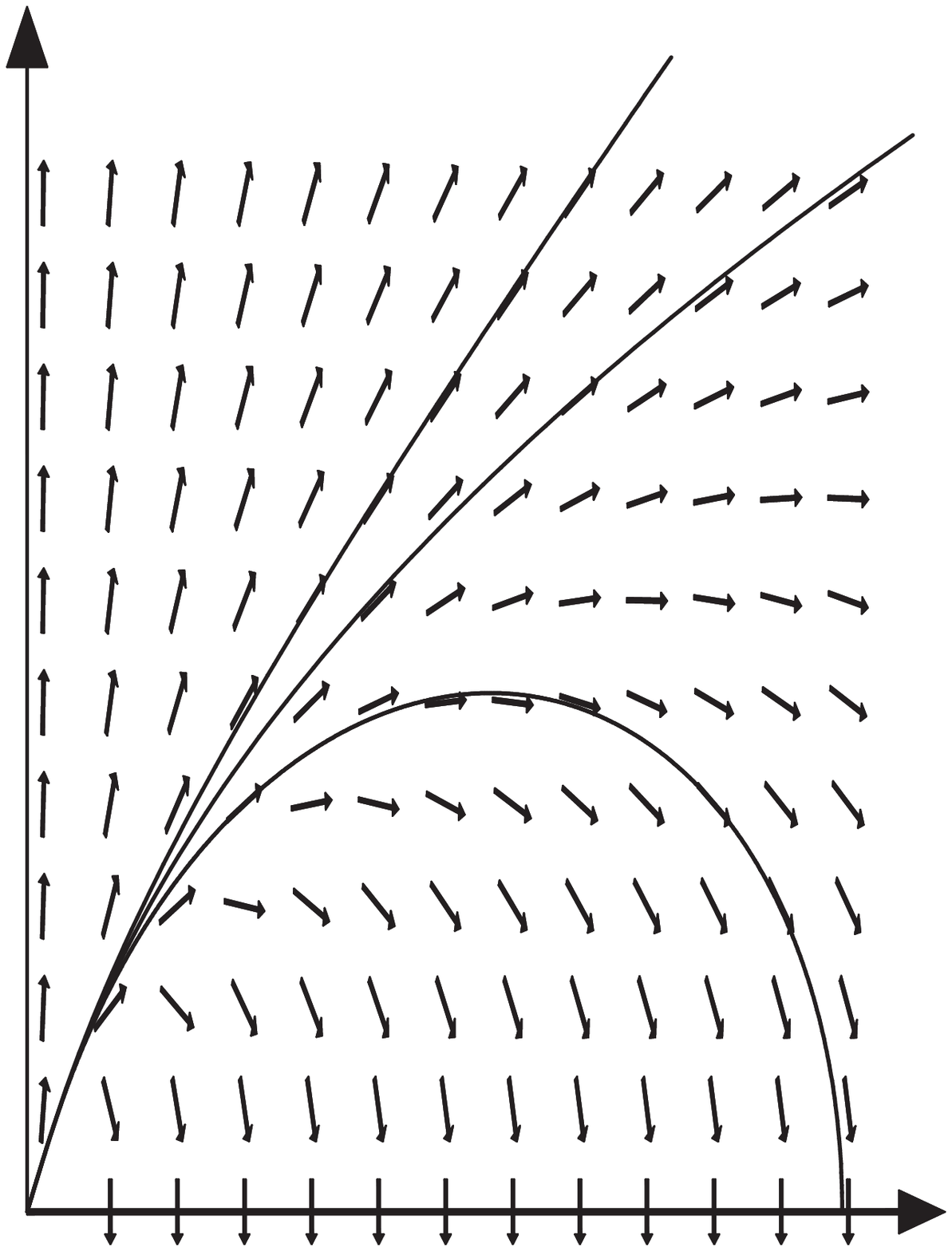,height=8.5cm}}
\put( 90,0){\epsfig{file=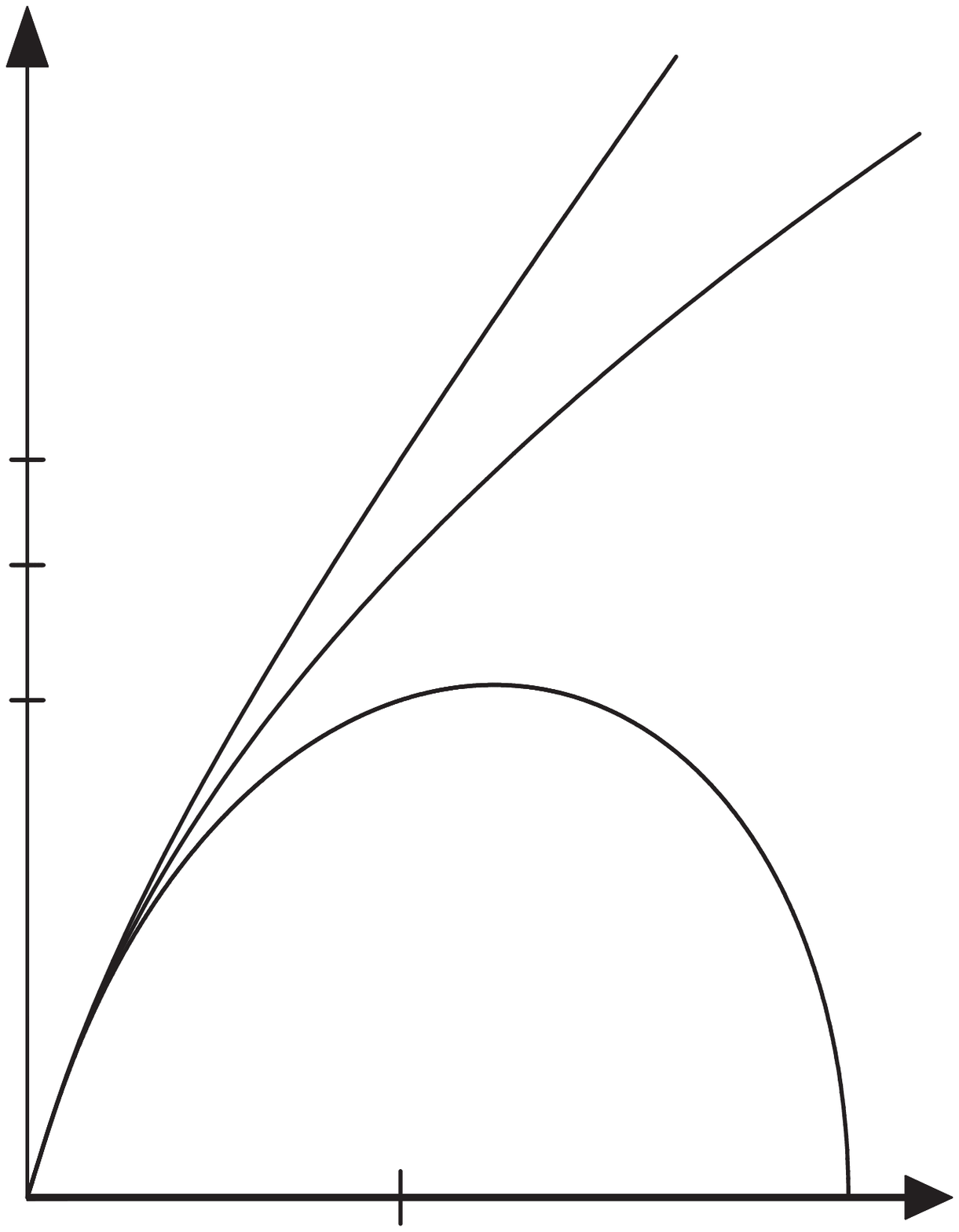,height=8.5cm}}
\put(147,-3){$x^{\star}$}
\put(154,4){$x$}
\put(72,4.5){$x$}
\put(116.5,-4){$1$}
\put(94,83){$\gamma_1$}
\put(80,45){$\gamma_1^{\star}(1)$}
\put(80,52.25){$\gamma_1(1)$}
\put(80,35.25){$\gamma_1(1)$}
\put(14,83){$\gamma_1$}
\end{picture}
\end{center}
\caption{$P(x)=x$, $s=1$ illustrating that all solutions of
(\ref{preDSeqn}) tend to $0$ as $x\to0$, but that some choices of
$\gamma_1(1)$ lead to solutions that extend as $x\to\infty$,
while others lead to $\gamma_1(x^{\star})=0$ at some finite
$x^{\star}$ and cease to exist beyond that point.}
\label{fig:solutionsnonsystem}
\end{figure}
Note first that
$\frac{\gamma_1(x)}{x}\to1$ as $x\to0$, irrespective of
$\gamma_1(1)$. On the other hand, A careful study of
(\ref{eqn:solsol}) shows that there is a preferred value
$\gamma_1^{\star}(1)=1+W(-\ed^{-2})=0.8414\ldots$ that separates
initial conditions at $x=1$ in two disjoint intervals, ${\rm
I}(1)=(0,\gamma_1^{\star}(1))$ and ${\rm
G}(1)=[\gamma_1^{\star}(1),\infty)$. All solutions with
$\gamma_1(1)\in{\rm G}(1)$ are global, i.e. they exist for all
$x\geq0$, while all solutions with $\gamma_1(1)\in{\rm I}(1)$
satisfy $\gamma_1(x^{\star})=0$ for some $x^{\star}>1$ 
depending on $\gamma_1(1)$ and cannot be continued
beyond $x=x^{\star}$. The solution $\gamma_1^{\star}(x)$ of
(\ref{preDSeqn}) with $\gamma_1(1)=\gamma_1^{\star}$ is thus the
smallest global solution and is called the separatrix for that
reason. Furthermore,
$\gamma_1^{\star}(x)\sim\sqrt{2x}-\frac{2}{3}+{\cal O}(x^{-1/2})$
as $x\to\infty$. See also figure \ref{fig:solutionsnonsystem}.

If $s=2$, on the other hand, one can only obtain an implicit
formula for the solutions, as was originally done in \cite{bkerfc}:
\begin{equs}
\sqrt{x}
=
\sqrt{x_0}~ 
\ed^{\Gamma_1(x_0)^2-\Gamma_1(x)^2}-
\sqrt{2}
\ed^{-\Gamma_1(x)^2}
\int_{\Gamma_1(x)}^{\Gamma_1(x_0)} \ed^{z^2}{\rm d}z
~~~\mbox{where}~~~
\Gamma_1(x)=\frac{1+\gamma_1(x)}{\sqrt{2x}}~.
\label{eqn:hihihi}
\end{equs}
Note that $\Gamma_1(x)\leq\Gamma_1(x_0)$ for all $x\geq x_0$ (as
long as $\gamma_1(x)\geq0$, or as long as $\gamma_1(x)$ exists),
since
\begin{equs}
\frac{{\rm d}\Gamma_1(x)}{{\rm d}x}=
-\frac{1}{2^{3/2}\sqrt{x}\gamma_1(x)}\leq0~.
\end{equs}
In particular, (\ref{eqn:hihihi}) gives an upper bound on the
maximal interval of existence of solutions:
\begin{equs}
x\leq x_0~\ed^{2\Gamma_1(x_0)^2}<\infty~.
\end{equs}
We thus see that if $P(x)=x$ and $s=2$, there are {\em no} global
solutions of (\ref{preDSeqn}).

Going beyond $P(x)=x$% and $s=1$ or $s=2$, one can try substituting
%a series expansion at $x=0$ into the differential equation
%(\ref{preDSeqn}) (assuming $P$ to be analytic, say), getting
%\begin{equs}
%\gamma_1(x)&=P'(0)x+\myl{12}
%\frac{P''(0)}{2}+(s-1)P'(0)^2
%\myr{12}x^2\\
%&\phantom{=~}+
%\myl{12}
%\frac{P'''(0)}{6}+
%({\textstyle\frac{3}{2}-1})
%P'(0)P''(0)+
%(s-1)(3s-2)P'(0)^3
%\myr{12}x^3+\ldots~.
%\end{equs}
%Assuming the series is convergent, we denote the corresponding
%solution by $\tilde{\gamma}_{1}(x)$. However, for any
%$C\in{\bf R}$, we note that
%\begin{equs}
%\gamma_{1}(x)=\tilde{\gamma}_{1}(x)\myl{12}
%1+W\myl{10}
%C\ed^{-\frac{1}{{sP'(0)x}}}
%\myr{10}
%\myr{12}
%\label{eqn:beyondallorders}
%\end{equs}
%also solves the series expansion of (\ref{preDSeqn}) at $x=0$, in
%agreement with the $P(x)=x$ and $s=1$ case, where we would have
%$\tilde{\gamma}_1(x)=x$, see also (\ref{eqn:solsol}).
%Unfortunately, this analysis, being perturbative in nature, cannot
%give any information on the question of global existence of
%solutions.
%
%To answer that question
, we first define the following (possibly
infinite) quantities
\begin{equs}
\label{condition_first}
{\cal D}_s(P)&=\int_{x_0}^\infty \frac{P(z)}{z^{1+2/s}} {\rm
d}z~~~~~\mbox{and}~~~~~
{\cal L}(P)=
\int_{x_0}^{\infty}
\frac{2{\rm d}z}{z(\sqrt{1+4P(z)}-1)}~.
\end{equs}
We can now state our rigorous results, which will be proven later
by a strategy largely inspired by \cite{vBSW}:
\begin{itemize}
\item Under the hypothesis H1 alone, there are no global solutions if
${\cal D}_s(P)=\infty$, while there exist some global solutions
if ${\cal D}_s(P)<\infty$. Note that in the $P(x)=x$ case, we
recover the previous analysis, since ${\cal D}_s(P(x)=x)<\infty$
if and only if $s<2$. 
\item Under the additional hypothesis H2, there is a (non-trivial)
minimal solution $\gamma_1^{\star}(x)$ which exists for all $x>0$
and separates global solutions (above $\gamma_1^{\star}(x)$) from
solutions that exist only for finite $x$ (below
$\gamma_1^{\star}(x)$).
\end{itemize}

The separatrix $\gamma_1^{\star}(x)$, in the case when it exists,
is thus the minimal physical solution, and it matches
perturbation theory near the origin. Further, with appropriate
conditions on $P(x)$, its behavior in terms of the running
coupling is extremely special as we will discuss below.

Consequently we conjecture that it is the solution chosen by nature.
Note that this solution does not give us a preferred value for $x$: if we vary $
x$ in accordance with  the renormalization group equation for the
running coupling,
\begin{equs}
\frac{{\rm d} x(x_0,L)}{{\rm d}L}=\beta(x(x_0,L))~,
\end{equs}
we just move along our distinguished curve, but there is no
preferred value of $ x$ from the existence of a distinguished
solution.

Following the proof of these results we interpret them in terms
of the running coupling, by using the renormalization group
equation, which in this case reduces to
\begin{equs}
\frac{{\rm d} x}{{\rm d}L}=\beta(x)=s x \gamma_1(x)~.
\end{equs}
If ${\cal D}_s(P)<\infty$, then we will show that if also ${\cal
  L}(P)<\infty$ then all global solutions
give Landau poles, whereas the separatrix
is the only global solution that does {\em not} lead to a Landau pole if
${\cal L}(P)=\infty$, or, in particular, if
$\lim_{x\to\infty}P(x)<\infty$.

The remainder of this paper is organized as follows: in Section
\ref{main}, we will consider the existence/absence of global
solutions of (\ref{preDSeqn}), prove the existence of the
separatrix in the appropriate case, and state the asymptotic
properties of the global solutions. Then, in Section
\ref{sec:running}, we will interpret the results of Section
\ref{main} in terms of the running coupling. The paper concludes
with Section \ref{sec:technical}, in which we give the details
left over in the proofs of Section \ref{main} and
\ref{sec:running}.

%Note that we can parameterize various curves for a non-vanishing
%$\beta$-function by their value taken at some positive $x_0>0$,
%or equivalently by their curvature at the origin. In
%perturbation theory this value is given just by the coefficient
%of the second derivative at zero, which amounts to the two-loop
%coefficient of the $\beta$-function, which in QED is a
%scheme-invariant for all renormalization schemes.

%We also note that the renormalization group equation for
%$x=x(x_0,t)$ above dictates that asymptotically $x$ has a
%horizontal tangential as a function of $t$ for solutions for
%$\beta$ below our distinguished curve (as these solutions for
%$\beta$ vanish at some finite positive $x_1>x_0$), while
%solutions above the separatrix develop a Landau pole for a
%finite (possibly very large) $x_1$, $x(x_1,t_1)=+\infty$ for
%finite $t_1$, as any approximation of $\beta$ by a polynomial
%forces such a pole to appear. We hence get a vertical
%tangential for such $x$.

%WHAT CAN WE SAY ABOUT THE SEPARATRIX ITSELF?

\section{Main results}
\label{main}

Building on our analysis of the $P(x)=x$, $s=1$ case, and since,
$f(\gamma_1,x)$ is singular at both $x=0$ and $\gamma_1=0$, we
first avoid those singularities by considering instead of
(\ref{preDSeqn}) the initial value problem
\begin{equs}
\label{DSeqn}
\frac{{\rm d}\gamma_1(x)}{{\rm d}x} =
\frac{\gamma_1(x)
+\gamma_1(x)^2-P(x)}{sx\gamma_1(x)}~,~~~
\gamma_1(x_0)=\gamma_0>0~,
\end{equs}
for some $x_0>0$. %Note that we could not instead use as initial
%data neither $\gamma_1(0)=0$ or $\gamma'(0)=P'(0)$, nor any
%higher derivative at $x=0$, see (\ref{eqn:beyondallorders}).
Since $f(\gamma_1,x)$ is regular away from $x=0$ and
$\gamma_1=0$, solutions of (\ref{DSeqn}) exist {\em locally}
around $x=x_0$. Furthermore these solutions are unique and are
continuous w.r.t. the initial condition $\gamma_0$. These three
statements (local existence, uniqueness and continuity) can be
rigorously proved with standard techniques, using that if
$\gamma_1$ and $\gamma_2$ are solutions of (\ref{DSeqn}), they
satisfy the integral equations
\begin{equs}
\label{solution}
\gamma_i(x) &= \left(\frac{x}{x_0}\right)^{1/s}(1+\gamma_i(x_0))-1
- x^{1/s}\int_{x_0}^x \frac{P(z)}{sz^{1+1/s}\gamma_i(z)}{\rm
d}z~,\\
\gamma_1(x)-\gamma_2(x)&=
\myl{12}
\gamma_1(x_0)-\gamma_2(x_0)
\myr{12}~
\exp\myl{18}
\int_{x_0}^{x}
\frac{1}{sz}
+\frac{P(z)}{sz\gamma_1(z)\gamma_2(z)}
{\rm d}z
\myr{18}~,
\label{eqn:diff}
\end{equs}
as long as they exist (and are strictly positive).

We now prove that
global solutions of (\ref{DSeqn}) exist if and only if
\begin{equs}
\label{condition}
\int_{x_0}^\infty  \frac{P(z)}{z^{1+2/s}} {\rm d}z <\infty~,
\end{equs}
for some finite $x_0>0$. Note that in the particular case
$P(x)=x$, (\ref{condition}) reduces to $s<2$, which agrees with
our previous analysis. This is the content of the following theorem:
\begin{theorem}\label{thm:dichotomy}
Let $s>0$ and $P$ be a ${\cal C}^2$ everywhere positive
function. There exist positive global solutions of (\ref{DSeqn})
if and only if $P$ satisfies the integrability condition
(\ref{condition}) for some $x_0>0$.
\end{theorem}

Before proving Theorem \ref{thm:dichotomy}, we note that
condition (\ref{condition}) places a strong restriction
on the asymptotic behavior of $P(x)$ as $x\to\infty$. For example
in the case of QED, $s=1$, $P(x)$ can grow at most like
$o(x^{2})$ as $x\to\infty$ for global positive solution of
(\ref{DSeqn}) to exist. On the other hand, if
$\lim_{x\to\infty}P(x)<\infty$, (\ref{condition}) is satisfied
{\em for all} $s>0$.

\begin{proof}
Consider first that (\ref{condition}) holds and let $x_0>0$.
Choose then 
\begin{equs}
\gamma_1(x_0)=
x_0^{1/s}
\left(
\frac{2}{s}\int_{x_0}^{\infty}\frac{P(z)}{z^{1+2/s}}{\rm d}z
+\epsilon^2\right)^{1/2}
\label{eqn:g0choice}
\end{equs}
for some $\epsilon>0$. Assume {\it ab absurdum} that the
corresponding solution $\gamma_1(x)$ has a maximal finite
interval of existence $[x_0,x_1]$ for some $x_1>x_0$. It follows
that either $\gamma_1(x_1)=\infty$ or $\gamma_1(x_1)=0$. The
first case cannot happen since from (\ref{solution}), we find
\begin{equs}
\gamma_1(x) \leq
\left(\frac{x}{x_0}\right)^{1/s}
\left(1+\gamma_1(x_0)\right)
\label{eqn:basic}
\end{equs}
for all $x\in[x_0,x_1]$, hence $\gamma_1(x_1)=0$. This also leads
to a contradiction, since rewriting (\ref{DSeqn}) as
\begin{equs}
\label{fthmest}
\frac{1}{2}\frac{{\rm d}}{{\rm d}x}(\gamma_1(x)^2) = \frac{\gamma_1(x)^2}{sx}
+\frac{\gamma_1(x)}{sx} -\frac{P(x)}{sx} \ge \frac{\gamma_1(x)^2}{sx}
-\frac{P(x)}{sx}
\end{equs}
(using $\gamma_1(x)\geq0$ for $x\in[x_0,x_1]$), integrating that
inequality on $[x_0,x_1]$ and using (\ref{eqn:g0choice}) gives
\begin{equs}
\gamma_1(x_1)^2 \ge x_1^{2/s}\left(\frac{\gamma_1(x_0)^2}{x_0^{2/s}}
-\frac{2}{s}\int_{x_0}^{x_1}\frac{P(z)}{z^{1+2/s}}{\rm d}z
\right)=\myl{10}\epsilon~x_1^{1/s}\myr{10}^2>0~.
\label{eqn:contradic}
\end{equs}
This contradicts our {\it ab absurdum} assumption that
$\gamma_1(x_1)=0$, and so $x_1=\infty$.

To prove the converse, assume {\it ab absurdum} that there
exist a global positive solution of (\ref{DSeqn}) for any
$\gamma_1(x_0)>0$ if
\begin{equs}
\lim_{x\to\infty}
\int_{x_0}^{x} \frac{P(z)}{z^{1+2/s}} {\rm d}z=\infty~.
\end{equs}
Since the solution is global, (\ref{eqn:basic}) holds for all
$x\geq x_0$, and inserting (\ref{eqn:basic}) into
(\ref{solution}) gives
\begin{equs}
\gamma_1(x) \le x^{1/s}\left( \frac{1+\gamma_1(x_0)}{x_0^{1/s}}
- \frac{x_{0}^{1/s}}{s(1+\gamma_1(x_0))} \int_{x_0}^x
\frac{P(z)}{z^{1+2/s}}{\rm d}z\right) -1~,
\label{eqn:improved}
\end{equs}
a contradiction, since (\ref{eqn:improved}) becomes negative as
$x\to\infty$.
\end{proof}

In the case where global positive solutions do exist, we now
prove that the notion of {\em smallest} global positive solution
(the separatrix) is well defined, at least if $P$ is strictly
increasing:
\begin{theorem}
\label{thm:separatrixintro}
Let $x_0>0$, $s>0$ and assume that $P$ is a ${\cal C}^2$
everywhere positive increasing function that satisfies
(\ref{condition}). Then there exist a unique value
$\gamma_1^{\star}(x_0)$ such that the solution of (\ref{DSeqn})
exists globally if and only if
$\gamma_1(x_0)\geq\gamma_1^{\star}(x_0)$. Furthermore, for every
global solution $\gamma_1(x)$, there exists a constant $C>0$ such
that for all $x\geq0$,
\begin{equs}
\gamma_c(x)<
\gamma_1^{\star}(x)\leq\gamma_1(x) \le \gamma_c(x) +
Cx^{\frac{1}{s}}
+
\my{\{}{18}
\begin{array}{ll}
0& \mbox{if }~~x\geq x_0 \\[1mm]
B_s(x,x_0) & \mbox{if }~~x\leq x_0
\end{array}
\label{eqn:refinedestimates}
\end{equs}
where $\gamma_1^{\star}(x)$ is the solution of (\ref{DSeqn}) that
corresponds to the initial condition $\gamma_1^{\star}(x_0)$ and
\begin{equs}
\gamma_c(x)&=\frac{\sqrt{1+4P(x)}-1}{2}~,
\label{eqn:defnull}
\\
B_s(x,x_0)&=
x^{1/s}\int_{x}^{x_0}
\frac{{\rm d}z}{z^{1/s}}=
\my{\{}{24}
\begin{array}{ll}
{\cal O}(x) & \mbox{as }x\to0~~\mbox{if }~~s<1 \\[1mm]
{\cal O}(x~|\ln(x)|) & \mbox{as }x\to0~~\mbox{if }~~s=1\\[1mm]
{\cal O}(x^{1/s}) & \mbox{as }x\to0~~\mbox{if }~~s>1
\end{array}~.
\label{eqn:def_fs}
\end{equs}
In particular, $\lim_{x\to0}\gamma_1(x)=0$ for every positive
global solution.
\end{theorem}

\begin{proof}
The technical details of the proof will be given in Section
\ref{sec:technical} below. We first note that solutions can have
at most one global maximum, and no local minima, and that the global
maximum can only occur on the nullcline $\gamma_c(x)$ as defined
in (\ref{eqn:defnull}). Namely, if $x^{\star}$ is an extremum,
then
\begin{equs}
\gamma_1(x^{\star})=\gamma_c(x^{\star})~,~~~\gamma_1'(x^{\star})=
0~~~\mbox{ and }~~~
\gamma_1''(x^{\star})=-\frac{P'(x^{\star})}{s x^{\star}\gamma_c(x^{\star})}
<0~.
\label{eqn:stuff}
\end{equs}
These relations have the following consequences. First, solutions
of (\ref{DSeqn}) with $\gamma_1(x_0)\geq\gamma_c(x_0)$ cannot
have a global maximum at some $x_1<x_0$, nor a local minimum at
such a point, and hence must decay monotonically to $0$ as
$x\to0$, while satisfying $\gamma_1(x)>\gamma_c(x)$ for all
$x\in[0,x_0]$. Second, solutions of (\ref{DSeqn}) with
$\gamma_1(x_0)<\gamma_c(x_0)$ will have a global maximum at some
$x_1<x_0$, and will then decay monotonically to $0$ as
$x\to0$, while satisfying $\gamma_1(x)>\gamma_c(x)$ for all
$x\in[0,x_1]$ by the above argument. In particular, all solutions
of (\ref{DSeqn}) can be continued as $x\to0$ and more refined
arguments (see Lemma \ref{lemma:continuation0}) show that they
satisfy (\ref{eqn:refinedestimates}) for all $x\in[0,x_0]$.

The relations (\ref{eqn:stuff}) also show that a
solution that satisfies $\gamma_1(x_0)<\gamma_c(x_0)$ must
decrease monotonically for all $x\geq x_0$, and more refined
arguments (see Lemma \ref{lemmaD}) will show that those
solutions indeed satisfy $\gamma_1(x_1)=0$ for some finite
$x_1>x_0$ and thus cannot be continued as $x\to\infty$.

Furthermore, since $\gamma_c$ is itself monotonically increasing,
solutions that start with $\gamma_1(x_0)=\gamma_c(x_0)+\epsilon$
with $\epsilon\ll1$ necessarily cross the nullcline at some
$x>x_0$, and thus also cannot be continued indefinitely as
$x\to\infty$, see Lemma \ref{lemmaC} below. On the other hand,
when (\ref{condition}) holds, Theorem \ref{thm:dichotomy} shows that
there are large enough initial conditions whose corresponding
solutions can be continued as $x\to\infty$, and thus never
cross the nullcline.

By the above arguments, continuity of solutions with respect to
initial conditions and equation (\ref{eqn:diff}),
\begin{equs}
{\rm I}(x_0) = \{\gamma_1(x_0) > \gamma_c(x_0) ~|~\exists x>
x_0~~\mbox{ with }~~~
\gamma_1(x) < \gamma_c(x) \}
\end{equs}
is a single open, bounded, non-empty interval. Define now
$\gamma_1^{\star}(x_0)$ as the supremum of ${\rm
I}(x_0)$. From (\ref{eqn:diff}), no solution starting below
$\gamma_1^{\star}(x_0)$ can exist globally, and all solutions
starting above must stay above the solution corresponding to
$\gamma_1^{\star}(x_0)$, and (\ref{eqn:basic}) then implies
(\ref{eqn:refinedestimates}) as $x\geq x_0$.
\end{proof}

We conclude this section with the following corollary about the
growth of global solutions as $x\to\infty$.

\begin{corollary}
\label{cor:growth}
Let $s>0$ and assume that $P$ satisfies (\ref{condition}).
Then every global solution of (\ref{DSeqn}) with
$\gamma_1(x_0)>\gamma_1^{\star}(x_0)$ satisfies $C_1
~x^{\frac{1}{s}}\leq\gamma_1(x)\leq
C_2~x^{\frac{1}{s}}$ as $x\to\infty$ for some $0<C_1<C_2$, while the
separatrix itself satisfies
\begin{equs}[2]
\gamma_c(x)<
\gamma_1^{\star}(x)&\leq
\min\myl{12}\lim_{x\to\infty}\gamma_c(x)~,~
C~x^{\frac{1}{s}}
\myr{12}
\end{equs}
for some $C>0$. In particular, if
$\displaystyle\lim_{x\to\infty}P(x)<\infty$, the separatrix is
the only global bounded solution of (\ref{DSeqn}).
\end{corollary}

\begin{proof}
Let $\gamma_1(x_0)>\gamma_{1}^{\star}(x_0)$, and consider the
corresponding solution of (\ref{DSeqn}). The upper bound
$\gamma_1(x)\leq C_2~x^{\frac{1}{s}}$ follows immediately from
(\ref{solution}). For the lower bound, we note that from
(\ref{eqn:diff}), we
have
\begin{equs}
\gamma_1(x)&\geq\gamma_1^{\star}(x)+
\myl{10}
\gamma_1(x_0)-\gamma_1^{\star}(x_0)
\myr{10}
~\exp\myl{12}
\int_{x_0}^{x}
\frac{{\rm d}z}{sz}
\myr{12}\geq C_1~x^{\frac{1}{s}}
~,
\end{equs}
for some $C_1>0$ since $\gamma_1(x_0)>\gamma_1^{\star}(x_0)$.

As for the separatrix itself, first note that the lower bound is
already contained in Theorem \ref{thm:separatrixintro}. If
$\lim_{x\to\infty}P(x)=\infty$, the upper bound
$\gamma_1^{\star}(x)\leq C~x^{\frac{1}{s}}$ follows again from
(\ref{solution}). If $\lim_{x\to\infty}P(x)<\infty$, we first set
$\gamma_{\infty}=\lim_{x\to\infty}\gamma_c(x)<\infty$. Consider
then $\gamma_1(x_0)=\gamma_{\infty}$. The corresponding solution
$\gamma_1(x)$ of (\ref{DSeqn}) must initially increase above
$\gamma_{\infty}$ for $x$ sufficiently close to $x_0$ since
\begin{equs}
\frac{{\rm d}\gamma_1}{{\rm d}x}
\my{|}{12}_{x=x_0}=\frac{\gamma_{\infty}+\gamma_{\infty}^2-P(x_0)
}{s x_0\gamma_{\infty}}
=\frac{\lim_{x\to\infty}P(x)-P(x_0)}{sx_0\gamma_{\infty}}>0~.
\end{equs}
Once the solution is above $\gamma_{\infty}$, it cannot have a
local maximum at an $x>x_0$ and hence can be continued as
$x\to\infty$. If $\lim_{x\to\infty}P(x)<\infty$, we thus have a
one parameter family of global solutions, indexed by $x_0$, the
point at which $\gamma_1(x_0)=\gamma_{\infty}$. Since the
separatrix $\gamma_{1}^{\star}$ is the smallest global solution,
we get $\gamma_1^{\star}(x)\leq\gamma_{\infty}$ for all $x>0$,
which concludes the proof.
\end{proof}

\section{The running coupling}
\label{sec:running}

We now interpret the above analysis in view of the running of the
coupling constant. With appropriate conventions this introduces
the second differential equation
\begin{equs}
  \frac{{\rm d} x }{{\rm d} L} = \beta(x(L))~.
    \label{eqn:easfofL}
\end{equs}
In the single equation case, combining (\ref{eqn:easfofL}) with
(\ref{DSeqn}), we obtain the following system
\begin{equs}
\frac{{\rm d} \gamma_1}{{\rm d} L} & = \gamma_1+\gamma_1^2- P~,~~~~
\frac{{\rm d} x }{{\rm d} L} = s~x~\gamma_1~,
\label{eqn:systemohmysystem}
\end{equs}
which we supplement with initial conditions at $L=0$:
\begin{equs}
x(L=0)=x_0~~~\mbox{and}~~~\gamma_1(L=0)=\gamma_1(x_0)~.
\end{equs}
Before considering the fate of non-global solutions of
(\ref{DSeqn}), we first explain how (almost all) global solutions
of (\ref{DSeqn}) are Landau poles.

\begin{theorem}
Assume that $P$ is a ${\cal C}^2$, positive, everywhere
increasing function that satisfies (\ref{condition}). The
separatrix $\gamma_1^{\star}$ is a Landau pole if and only if
\begin{equs}
{\cal L}(P)=
\int_{x_0}^{\infty}
\frac{{\rm d}z}{z~\gamma_c(z)}=
\int_{x_0}^{\infty}
\frac{2{\rm d}z}{z(\sqrt{1+4P(z)}-1)}<\infty~.
\end{equs}
All other global solutions of (\ref{DSeqn}) are Landau poles,
irrespective of the value of ${\cal L}(P)$.
\end{theorem}

\begin{proof}
We first note that global solutions of (\ref{DSeqn}) give solutions of
(\ref{eqn:systemohmysystem}) via the reparametrization
\begin{equs}
L=\int_{x_0}^{x(L)}
{\frac {{\rm d}z}{s~z~\gamma_1(z)}}~.
\end{equs}
In particular, a global solution of (\ref{eqn:systemohmysystem}) with
$x(L=0)=x_0$ and $\gamma_1(L=0)=\gamma_1(x_0)$ reaches $x=\infty$
at
\begin{equs}
L^{\star}=\int_{x_0}^{\infty}
{\frac {{\rm d}z}{s~z~\gamma_1(z)}}~.
\label{eqn:popoles}
\end{equs}
From Corollary (\ref{cor:growth}), we know that any global
solution of (\ref{DSeqn}) that is not the separatrix grows at
least like $x^{\frac{1}{s}}$ as $x\to\infty$. In particular, 
the integral in the r.h.s.~of (\ref{eqn:popoles}) converges to
some finite $L^{\star}$, and $\gamma_1(L)$ diverges as $L\to
L^{\star}$, signaling that this solution is a Landau
pole. By Corollary \ref{cor:growth}, the separatrix is
also a Landau pole if $P(x)$ grows fast enough as $x\to\infty$ so
that ${\cal L}(P)<\infty$.

If $\lim_{x\to\infty}P(x)<\infty$, then Corollary
(\ref{cor:growth}) shows that
$\gamma_1^{\star}\leq\lim_{x\to\infty}\gamma_c(x)<\infty$, which
makes the integral in the r.h.s.~of (\ref{eqn:popoles})
divergent. In particular, the separatrix is the only global
solution of (\ref{DSeqn}) that {\em is not} a Landau pole when
written in terms of the running coupling $L$. In section
\ref{sec:technical} below, we will show that this actually holds
not only if $\lim_{x\to\infty}P(x)<\infty$ but also for all $P$
that grow sufficiently slowly as $x\to\infty$ so that ${\cal
L}(P)=\infty$.
\end{proof}

Consider now $\gamma_1(x)$ a solution of (\ref{DSeqn}) that only
exist on maximal finite interval $x\in[0,x^{\star}]$. By the
results of section \ref{main}, we necessarily have
$\gamma_1(x^{\star})=0$. As is apparent from
(\ref{eqn:systemohmysystem}), the introduction of the running
coupling removes the singularity of (\ref{DSeqn}) at
$\gamma_1=0$. There is thus a 1-1 correspondence between
solutions of (\ref{DSeqn}) that exist only on finite intervals
with the family of solutions of (\ref{eqn:systemohmysystem}) with
$x(L=0)=x^{\star}$ and $\gamma_1(L=0)=0$. More precisely, we have
the
\begin{theorem}
For each $x^{\star}>0$, there is a unique solution of
(\ref{eqn:systemohmysystem}) that satisfies $x(L=0)=x^{\star}$
and $\gamma_1(L=0)=0$. This solution is an heteroclinic orbit
of the system (\ref{eqn:systemohmysystem}) connecting the two
equilibrium points $(x,\gamma_1)=(0,0)$ at $L=-\infty$ to
$(x,\gamma_1)=(0,-1)$ at $L=\infty$.
\end{theorem}
Note that this theorem implies that solutions of (\ref{DSeqn})
that exist only on finite intervals are actually double-valued
as functions of $x$, but exists for all $L\in{\bf R}$. Note that
in particular, such solutions come back to $x=0$ as a
dipole-ghost \cite{Nakanishi}: we gained a full integer in
scaling weight for the photon, see left panel of figure
\ref{fig:aroundthewolrdineightydays}.

\begin{proof}
Fix $x^{\star}>0$, and consider the solution of
(\ref{eqn:systemohmysystem}) that satisfies $x(L=0)=x^{\star}$ and
$\gamma_1(L=0)=0$. Note first that the vector field associated
with (\ref{eqn:systemohmysystem}) is perpendicular to the
$x$-axis, and crosses the $x=x^{\star}$ vertical line from left
to right above the $x$-axis, see also
the right panel of figure \ref{fig:aroundthewolrdineightydays}.
As a consequence, and by local existence of solutions of 
(\ref{eqn:systemohmysystem}), there exists a {\em finite} $L^{-}<0$
such that $\gamma_1(L^{-})>0$ and $0<x(L^{-})<x^{\star}$. By the
results of section (\ref{main}), the solution of (\ref{DSeqn})
with $x_0=x(L^{-})>0$ and $\gamma_1(x_0)=\gamma_1(x(L^{-}))>0$
can be extended up to $x=0$ and satisfies
$\gamma_1(x)\sim\gamma_c(x)$ as $x\to0$. In particular, the
solution of (\ref{eqn:systemohmysystem}) satisfying
$x(L=0)=x^{\star}$ and $\gamma_1(L=0)=0$ tends to
$(x,\gamma_1)=(0,0)$ as $L\to-\infty$ since
\begin{equs}
L=L^{-}-\int_{x(L)}^{x(L^{-})}
\frac{{\rm d}z}{s~z~\gamma_1(z)}
\to-\infty~~~\mbox{as}~~~x(L)\to0~.
\end{equs}
We now prove that $(x,\gamma_1)\to(0,-1)$ as $L\to\infty$.
Again, since the vector field associated with
(\ref{eqn:systemohmysystem}) is perpendicular to the
$x$-axis, and crosses the $x=x^{\star}$ vertical line
from right to left below the $x$-axis, there exists a finite
$L^{+}>0$ such that $-1<\gamma_1(L^{+})<0$ and
$0<x(L^{+})<x^{\star}$ (the value $\gamma_1(L^{+})$ is the dashed
line on the right panel of figure
\ref{fig:aroundthewolrdineightydays}). Note then that the vector
field points inside the rectangle
$R=[0,x^{\star}]\times[\gamma_1(L^{+}),-1-\gamma_c(x^{\star})]$,
except on the $\gamma_1$ axis where it is tangent and points
towards $(x,\gamma_1)=(0,-1)$, see also the right panel of figure
\ref{fig:aroundthewolrdineightydays}. It thus follows that
\begin{equs}
0\leq x(L)\leq x^{\star}~~~\mbox{and}~~~
-1-\gamma_c(x^{\star})\leq\gamma_1(L)\leq
\gamma_1(L^{+})~~~\forall L\geq{L^{+}}~.
\end{equs}
In particular,
\begin{equs}
-c_1~x\equiv
-sx(1+\gamma_c(x^{\star}))\leq \frac{{\rm d}x}{{\rm d}L}
\leq sx\gamma_1(x(L^{+}))
\equiv-c_2~x
~,
\end{equs}
for some $c_1,c_2>0$, and thus
\begin{equs}
x(L^{+})\ed^{-c_1(L-L^{+})}
\leq x(L)\leq
x(L^{+})\ed^{-c_2(L-L^{+})}~,
\end{equs}
which shows that $x(L)\to0$ as $L\to\infty$. Since the vector
field points towards $(x,\gamma_1)=(0,-1)$ on the 
$\gamma_1$ axis, it also follows that $\gamma_1(L)\to-1$ as
$L\to\infty$, and concludes the proof.
\end{proof}

\def\mystyle{\scriptstyle}
\begin{figure}
\begin{center}
\unitlength1mm
\begin{picture}(190,90)(0,0)
\put( 20,0){\epsfig{file=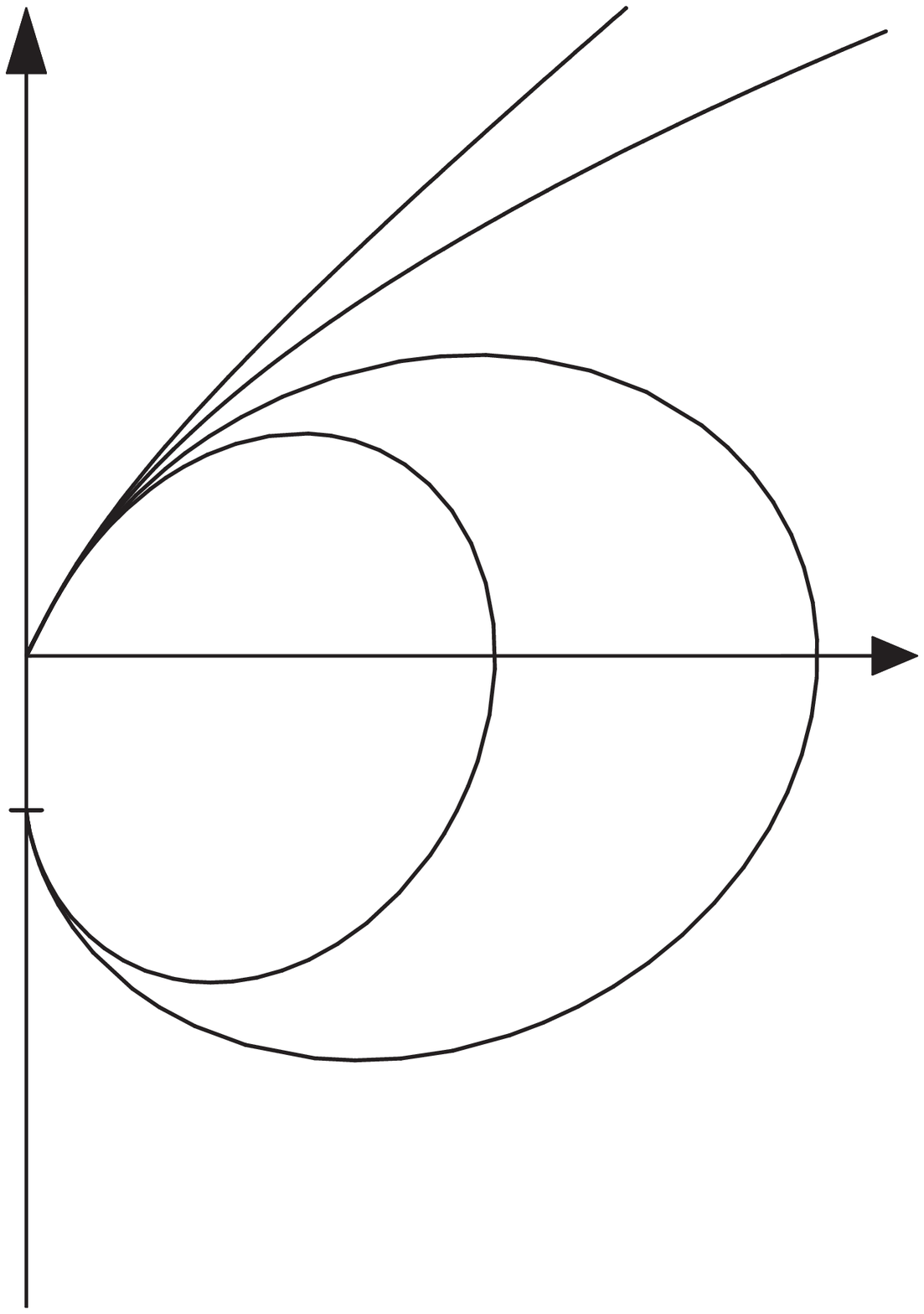,height=8.5cm}}
\put(100,0){\epsfig{file=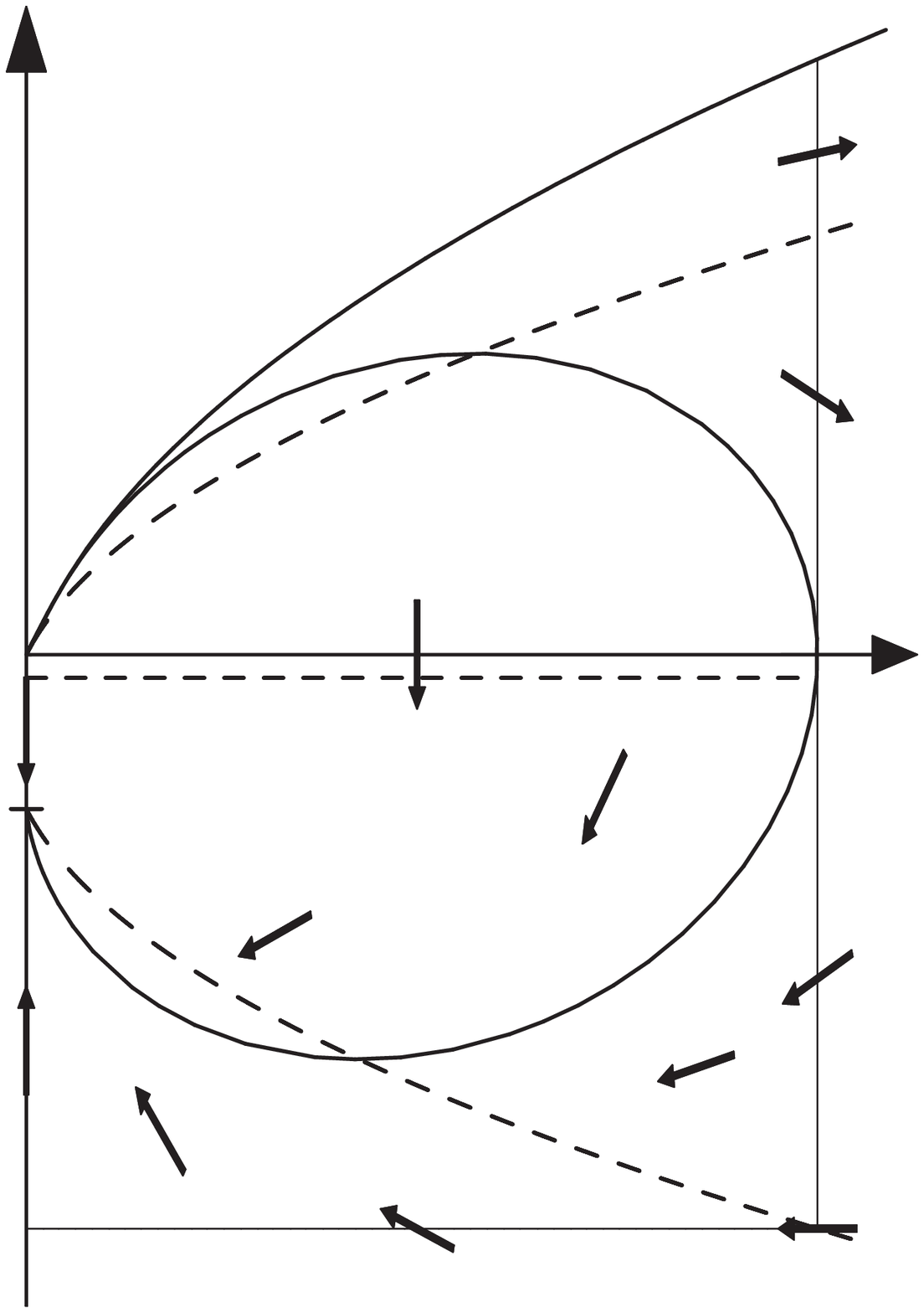,height=8.5cm}}
\put(153.5,44){$\mystyle x^{\star}$}
\put(73.5,44){$\mystyle x^{\star}$}
\put(158,39.5){$x$}
\put(78,39.5){$x$}
\put(95,83){$\gamma_1$}
\put(15,83){$\gamma_1$}
\put(136,79){$\mystyle\gamma_1^{\star}(x)$}
\put(156,72){$\mystyle\gamma_c(x)$}
\put(156,2){$\mystyle-1-\gamma_c(x)$}
\put(85,4.5){$\mystyle-1-\gamma_c(x^{\star})$}
\put(65,74){$\mystyle\gamma_1^{\star}(x)$}
\put(96,31.5){$\mystyle-1$}
\put(91,40.5){$\mystyle\gamma_1(L^{+})$}
\put(16,31.5){$\mystyle-1$}
\end{picture}
\end{center}
\caption{$P(x)=x$, $s=1$ illustrating that, as a function of $L$,
non-global solutions of (\ref{DSeqn}) turn around and head to
$-1$ as $L\rightarrow\infty$.}
\label{fig:aroundthewolrdineightydays}
\end{figure}

\section{Technical proofs}
\label{sec:technical}

This section contains the technical details needed for a complete
proof of Theorem \ref{thm:separatrixintro} above. Throughout this
section, we assume $P$ is a ${\cal C}^2$, positive, strictly
increasing function of $x$, which satisfies (\ref{condition}).

Our first step is to show that solutions that start below the
nullcline $\gamma_c(x_0)$ cannot be continued as $x\to\infty$. Note
that this does not follow directly from (\ref{eqn:stuff}), since
$\gamma_1(x)$ could a priori decrease indefinitely as
$x\to\infty$ without ever reaching $\gamma_1=0$.
\begin{lemma}\label{lemmaD}
Let $\gamma_1(x_0)<\gamma_c(x_0)$ then the solution of
(\ref{DSeqn}) satisfies $\gamma_1(x_1)=0$ for some finite $x_1>x_0$.
\end{lemma}
\begin{proof}
Let $\gamma_1(x_0)\equiv\gamma_c(x_0)-\epsilon$ for some
$0<\epsilon<\gamma_c(x_0)$. We first note
that $\gamma_1(x)\leq \gamma_1(x_0)$ for all $x\geq x_0$ such
that the solution exists, otherwise there would be a local
minimum at some $x^{\star}\in[x_0,x]$, which is precluded by
(\ref{eqn:stuff}). Since $P(x)$ is increasing, we find
\begin{eqnarray}\label{Dineq}
\frac{{\rm d}\gamma_1(x)}{{\rm d}x} & \le & \frac{\gamma_c(x_0) -\epsilon +
(\gamma_c(x_0) -\epsilon)^2 - P(x_0)}{s x (\gamma_c(x_0)
-\epsilon)} \nonumber\\
& \le & -\frac{\epsilon(1+2\gamma_c(x_0)
-\epsilon)}{sx(\gamma_c(x_0) -\epsilon)} \equiv -
\frac{R(x_0,\epsilon)}{x},
\end{eqnarray}
for some $R(x_0,\epsilon)>0.$ Integrating (\ref{Dineq}) on
$[x_0,x]$ gives
\begin{equs}
\gamma_1(x) \le \gamma_1(x_0) - R(x_0,\epsilon) \int_{x_0}^x
\frac{{\rm d}z}{z} = \gamma_c(x_0)-\epsilon -
R(x_0,\epsilon)\ln\left(\frac{x}{x_0}\right)~,
\end{equs}
which shows that $\gamma_1(x_1)=0$ for some $x_1\leq x_0
\exp\left(\frac{ \gamma_c(x_0) -\epsilon
}{R(x_0,\epsilon)}\right)<\infty$ as claimed.
\end{proof}
Our next step is to show that solutions that start close enough,
but above the nullcline at $x_0$ cross the nullcline at some
$x>x_0$, and thus cannot be continued as $x\to\infty$ by Lemma
\ref{lemmaD}.

\begin{lemma}\label{lemmaC}
Assume $\gamma_1(x_0) = \gamma_c(x_0) + \delta^2$. There exist
$\delta$ $>0$ sufficiently small such that if $\gamma_1(x)$
solves (\ref{DSeqn}), then $\gamma_1(x_0 +\delta) <
\gamma_c(x_0+\delta)$.
\end{lemma}
\begin{proof}
First note that by (\ref{solution}), $\gamma_1(x) \leq
(\frac{x}{x_0})^{1/s}(1+\gamma_1(x_0))-1$.
Under the assumption that $\gamma_1(x_0) = \gamma_c(x_0) +
\delta^2$, we thus find
\begin{equs}
\sup_{x\in [x_0,x_0+\delta]} \gamma_1(x) \le
\gamma_c(x_0)+C\delta
\end{equs}
for some constant $C=C(x_0,s)>0$. We now use the following estimate on
$f(\gamma_1(x),x)$
\begin{equs}
\sup_{x\in [x_0,x_0+\delta]} f(\gamma_1(x),x) &\leq
\frac{\gamma_c(x_0)
+C\delta+ (\gamma_c(x_0)
+C\delta)^2 - P(x_0)}{s x_0 (\gamma_c(x_0)
+C\delta)} \\
& \leq \frac{C\delta}{s x_0}\left(2
+\frac{1}{\gamma_c(x_0)+C\delta} \right) \le
M\delta
\end{equs}
for some constant $M=M(x_0,s)>0$. We thus find, upon
integration of (\ref{DSeqn}) that
\begin{equs}\label{lowerC}
\gamma_1(x_0+\delta) \le \gamma_c(x_0)
+(1+M)\delta^2.
\end{equs}
Now by Taylor's theorem, there exists a constant $N(x_0)$ such
that
\begin{equs}\label{upperC}
\gamma_c(x_0+\delta) \ge \gamma_c(x_0) +\gamma_c'(x_0)\delta+N(x_0)\delta^2~.
\end{equs}
Since $\gamma_c'(x_0)>0$, we can
choose $\delta$ sufficiently small so that
\begin{equs}
(1+M-N(x_0))\delta^2
< \gamma_c'(x_0)\delta~,
\end{equs}
which implies that $\gamma_1(x_0+\delta)<\gamma_c(x_0+\delta)$
and completes the proof.
\end{proof}

Our next step is to show that every local solution of
(\ref{DSeqn}) can be continued as $x\to0$. We will also show that
all solutions behave asymptotically like $\gamma_c(x)$ as
$x\to0$.

\begin{lemma}\label{lemma:continuation0}
Let $\gamma_1(x)$ be a (local) solution of (\ref{DSeqn}) with
$\gamma_1(x_0)>0$. Then that solution can be continued for all
$x\in[0,x_0]$. Furthermore, there exist $0<x_1\leq x_0$ 
(with $x_1=x_0\Leftrightarrow\gamma_1(x_0)\geq\gamma_c(x_0)$)
and a constant $C>0$ such that $\gamma_c(x)<\gamma_1(x) \le \gamma_c(x)+
Cx^{\frac{1}{s}}+C B_s(x,x_1)$ for all $x\in[0,x_1]$, where
$B_s(x,x_1)$ is defined in (\ref{eqn:def_fs}).
\end{lemma}

\begin{proof}
We first note that $\gamma_1(x_0)>0$ guarantees that the solution
exists locally around $x_0$. We now prove that it satisfies
\begin{equs}
\min(\gamma_c(x),\gamma_1(x_0))
\leq\gamma_1(x)\leq
\max(\gamma_c(x_0),\gamma_1(x_0))~~~\forall x\in[0,x_0]~,
\label{eqn:frofro}
\end{equs}
and hence can be continued up to $x=0$.
Recall from (\ref{eqn:stuff}) that a solution can have at most
one global maximum, and no local minimum. We now consider two
cases, $\gamma_1(x_0)\geq\gamma_c(x_0)$ and
$\gamma_1(x_0)<\gamma_c(x_0)$.

In the first case, we claim that
\begin{equs}
\gamma_1(x_0)\geq\gamma_c(x_0)~~~\Rightarrow~~~
\gamma_c(x)\leq\gamma_1(x)\leq\gamma_1(x_0)~~~
\forall x\in[0,x_0]~.
\label{eqn:frifri}
\end{equs}
Namely, if $\gamma_1(x_0)\geq\gamma_c(x_0)$, then $\gamma_1(x)$
must decrease as $x$ decreases, at least for all $x$ sufficiently
close to $x_0$ with $x<x_0$. This follows because either
$\gamma_1'(x_0)>0$ if $\gamma_1(x_0)>\gamma_c(x_0)$ or
$\gamma_1'(x_0)=0$ and $\gamma_1''(x_0)<0$ if
$\gamma_1(x_0)=\gamma_c(x_0)$). We thus get that
$\gamma_1(x)<\gamma_1(x_0)$ holds for all $0\leq x<x_0$ since the
solution cannot have a local minimum. Also, the solution cannot
have a maximum at $x_1<x_0$ either, since at that maximum,
$\gamma_1(x_1)<\gamma_1(x_0)$, which would require a
local minimum at some intermediate value $x^{\star}\in(x_1,x_0)$,
and hence $\gamma_1(x)>\gamma_c(x)$ for all $x<x_0$.

In the case $\gamma_1(x_0)<\gamma_c(x_0)$, we claim that there
exist $0<x_1<x_0$ such that $\gamma_1(x_1)=\gamma_c(x_1)$, or, in
other words, the solution crosses the nullcline at some
$x_1<x_0$. Namely, the solution must increase initially as $x$
decreases (since $\gamma_1'(x_0)<0$). Since the solution cannot
have a local minimum, $\gamma_1(x)\geq\gamma_1(x_0)$ as long as
it is below the nullcline, and hence it must cross the nullcline
(and have a global maximum) at some
$x_1\in(\gamma_c^{-1}(\gamma_1(x_0)),x_0)$. In
particular, the global maximum $\gamma_1(x_1)$ satisfies
$\gamma_1(x_1)\leq\gamma_c(x_0)$ since $\gamma_c$ is strictly
increasing. We thus find
\begin{equs}
\gamma_1(x_0)\leq\gamma_1(x)\leq\gamma_c(x_0)~~~\forall x\in[x_1,x_0]~.
\end{equs}
Since $\gamma_1(x_1)=\gamma_c(x_1)$, we apply (\ref{eqn:frifri})
with $x\in[0,x_1]$ and get (using also
$\gamma_c(x_1)\leq\gamma_c(x_0)$) that
\begin{equs}
\gamma_c(x)\leq\gamma_1(x)\leq\gamma_1(x_1)\leq\gamma_c(x_0)~~~
\forall x\in[0,x_1]~.
\end{equs}
This completes the proof of (\ref{eqn:frofro}).

Note now that in all cases, there exists $x_1\leq x_0$ such that
$\gamma_1(x)\geq\gamma_c(x)$ for all $x\in[0,x_1]$. In
particular, $\frac{P(z)}{\gamma_1(z)}\leq
\frac{P(z)}{\gamma_c(z)}=1+\gamma_c(z)$ for all $z\in[0,x_1]$.
Let now $x\in[0,x_1]$. From (\ref{solution}), we find
\begin{equs}
\label{lemmaBest1}
\gamma_1(x) \le
\left(\frac{x}{x_1}\right)^{1/s}(1+\gamma_1(x_1))-1 +
x_1^{1/s}\int_{x}^{x_1} \frac{1+\gamma_c(z)}{sz^{1+1/s}}{\rm d}z~,
\end{equs}
which, after integrating by parts, gives
\begin{equs}\label{IBP}
\gamma_1(x) \le
\left(\frac{x}{x_1}\right)^{1/s}(\gamma_1(x_1)-\gamma_c(x_1)) +
\gamma_c(x) + x^{1/s}\int_{x}^{x_1}
\frac{\gamma_c'(z)}{z^{1/s}}{\rm d}z~.
\end{equs}
Since $\gamma_c'(z)\leq C$ for all $z\in[0,x_1]$, we get
$\gamma_1(x)\leq\gamma_c(x)+Cx^{\frac{1}{s}}+C B_s(x,x_1)$
for all $x\in[0,x_1]$, which completes the proof.
\end{proof}

We now have all the tools to prove Theorem
\ref{thm:separatrixintro}, which we restate now in the form
\begin{theorem}
\label{thm:separatrix}
Assume that (\ref{condition}) holds. The set
\begin{equs}
{\rm I}(x_0) = \{\gamma_1(x_0) > \gamma_c(x_0) ~|~\exists x>
x_0~~\mbox{ with }~~~
\gamma_1(x) < \gamma_c(x) \}~,
\end{equs}
is a single, open, non-empty and bounded interval.
Moreover the solution $\gamma_1^{\star}(x)$ of (\ref{DSeqn}) with
\begin{equs}
\gamma_1^{\star}(x_0) = \sup({\rm I}(x_0))
\end{equs}
is the smallest solution that exists for all $x\in[0,\infty)$, and its graph
defines the separatrix, in the sense that any global solution
$\gamma_1(x)$ of (\ref{DSeqn}) satisfies
\begin{equs}
\gamma_c(x)<\gamma_{1}^{\star}(x)\leq
\gamma_1(x)\leq
\gamma_c(x)+
Cx^{1/s}+
C
\my{\{}{12}
\begin{array}{ll}
B_s(x) & \mbox{if }~~x\leq x_0 \\[1mm]
0 & \mbox{if }~~x>x_0
\end{array}
\label{eqn:finalbound}
\end{equs}
for all $x\in[0,\infty)$.
\end{theorem}

\begin{proof}
We first note that by Lemma \ref{lemmaC}, ${\rm
I}(x_0)\neq\emptyset$, and by Lemma \ref{lemmaD} solutions that
start in ${\rm I}(x_0)$ cannot be continued as $x\to\infty$.
Since global solutions exist by Theorem \ref{thm:dichotomy} for
all $\gamma_1(x_0)$ large enough, ${\rm I}(x_0)$ is bounded
above. Also, ${\rm I}(x_0)$ is open by continuity of solutions
with respect to initial conditions. Consider now $\gamma_1(x)$
and $\gamma_2(x)$, to be two solutions of (\ref{DSeqn}), for
which $\gamma_c(x_0)<\gamma_2(x_0)<\gamma_1(x_0)$ and
$\gamma_1(x_0)\in{\rm I}(x_0)$. We now claim that $\gamma_2(x_0)$
must also be in ${\rm I}(x_0)$. Namely, since
$\gamma_1(x_0)\in{\rm I}(x_0)$, there must be an $x_1>x_0$ such
that $\gamma_1(x_1)<\gamma_c(x_1)$. By (\ref{eqn:diff}), we have
$\gamma_2(x)<\gamma_1(x)$ as long as both solutions exist. In
particular, either $\gamma_2(x)$ exists on $[x_0,x_1]$, and
(\ref{eqn:diff}) shows that
$\gamma_2(x_1)<\gamma_1(x_1)<\gamma_c(x_1)$ and thus
$\gamma_2(x_0)\in{\rm I}(x_0)$, or $\gamma_2(x)$ cannot be
continued up to $x=x_1$ and since it cannot diverge to infinity by
(\ref{eqn:basic}), we must have $\gamma_2(x_2)=0$ for some
$x_2<x_1$, which also implies that $\gamma_2(x_0)\in{\rm
I}(x_0)$. This shows that ${\rm I}(x_0)$ is a single open
interval. We thus define
\begin{equs}
\gamma_1^{\star}(x_0) = \sup({\rm I}(x_0))~.
\end{equs}
Evidently, $\gamma_1^{\star}(x_0)\notin{\rm I}(x_0)$, and so the
corresponding solution $\gamma_1^{\star}(x)$ of (\ref{DSeqn}) satisfies
$\gamma_c(x)<\gamma_1^{\star}(x)\leq Cx^{1/s}$ for all $x \ge x_0$
(the lower bound follows by definition of ${\rm I}(x_0)$, the
upper bound by (\ref{eqn:basic})), while Lemma
\ref{lemma:continuation0} shows that $\gamma_1^{\star}(x)$ exists
for all $x\in[0,x_0]$ and satisfies
$\gamma_c(x)<\gamma_1^{\star}(x)\le\gamma_c(x)+
Cx^{1/s}+C B_s(x,x_0)$ for all $x\in[0,x_0]$. Using
(\ref{eqn:diff}) and (\ref{eqn:basic}) again shows that the
solution corresponding to any
$\gamma_1(x_0)\geq\gamma_1^{\star}(x_0)$ can also be continued as
$x\to\infty$ and satisfies (\ref{eqn:finalbound}).
\end{proof}

We conclude this section with a last result concerning the growth
of the separatrix in the case where $P$ is a slowly increasing
function.

\begin{lemma}
Assume that $P$ satisfies
\begin{equs}
\int_{x_0}^{\infty}
\frac{{\rm d}z}{z~\gamma_c(z)}=
\int_{x_0}^{\infty}
\frac{2{\rm d}z}{z\sqrt{1+4P(z)}-1}=\infty~,
\label{eqn:repopolre}
\end{equs}
for some $x_0>0$. Then there exists a constant $C>0$ such that
the separatrix $\gamma_1^{\star}$ satisfies
$\gamma_c(x)<\gamma_1^{\star}(x)\leq \gamma_c(x)+C\ln(x)$ as
$x\to\infty$. In particular,
\begin{equs}
\int_{x_0}^{\infty}
\frac{{\rm d}z}{z~\gamma_1^{\star}(z)}=\infty~.
\end{equs}
\end{lemma}

\begin{proof}
We first derive from (\ref{eqn:repopolre}) the following bounds
on the asymptotics of $P(x)$ and $P'(x)$ as $x\to\infty$:
\begin{equs}
P(x)<C_1\ln(x)^4~~~\mbox{and}~~~
\frac{P'(x)x}{\sqrt{1+4P(x)}}\leq C_2\ln(x)~.
\end{equs}
The first one is obvious, and assuming
$\frac{P'(x)x}{\sqrt{1+4P(x)}}>C_2\ln(x)$ gives $P(x)\geq
C_1(1+\ln(x))^4$ which contradicts (\ref{eqn:repopolre}). In
particular, (\ref{condition}) is satisfied for all $s>0$, and we
are guaranteed by Theorem \ref{thm:separatrix} that the
separatrix $\gamma_1^{\star}$ exists and satisfies
$\gamma_1^{\star}(x)>\gamma_c(x)$ for all $x>0$.

Let now $\overline{\gamma_{1}}(x)=\gamma_c(x)+\overline{C}\ln(x)$
for some $\overline{C}>0$, and consider the one parameter family
of solutions of (\ref{DSeqn}) obtained by fixing
$\gamma_1(x_0)=\overline{\gamma_{1}}(x_0)$ for different $x_0>1$.
At least, those solutions that start with $x_0$ sufficiently large
can be extended as $x\to\infty$, since 
\begin{equs}
\frac{{\rm d}}{{\rm d}x}\overline{\gamma_{1}}(x_0)
<
\frac{{\rm d}}{{\rm d}x}\gamma_{1}(x_0)
=\frac{\overline{\gamma_{1}}(x_0)+\overline{\gamma_{1}}(x_0)^2-P(
x_0)}{
sx_0\overline{\gamma_{1}}(x_0)}~,
\end{equs}
because 
\begin{equs}
\frac{{\rm d}}{{\rm d}x}\overline{\gamma_{1}}(x_0)&=
\frac{\overline{C}+C_2\ln(x_0)}{x_0}~~~\mbox{as}~~~x_0\to\infty\\
\frac{\overline{\gamma_{1}}(x_0)+\overline{\gamma_{1}}(x_0)^2-P(
x_0)}{
sx_0\overline{\gamma_{1}}(x_0)}
&=\frac{2\overline{C}\ln(x_0)}{s~x_0}
+{\cal O}(x_0^{-1})~.
\end{equs}
This shows that a solution that starts on the curve
$\overline{\gamma_{1}}(x)=\gamma_c(x)+\overline{C}\ln(x)$ at
$x=x_0$ cannot cross it at any $x$ with $x>x_0$, and thus is
a global solution. In particular, the separatrix must be smaller
than any of these solutions, and we get
$\gamma_1^{\star}(x)\leq\gamma_c(x)+\overline{C}\ln(x)$. From
this last estimate, we get immediately that
\begin{equs}
\int_{x_0}^{\infty}
\frac{{\rm d}z}{z\gamma_1^{\star}(z)}\geq
\frac{1}{2}\min\myl{12}
\int_{x_0}^{\infty}
\frac{{\rm d}z}{z\gamma_c(z)}
~,~
\int_{x_0}^{\infty}
\frac{{\rm d}z}{z\overline{C}\ln(z)}
\myr{12}=\infty~,
\end{equs}
which completes the proof.
\end{proof}

\bibliographystyle{gvb}
\markboth{\sc \refname}{\sc \refname}
\bibliography{refs,main,dirkbib}

\def\Rom#1{\uppercase\expandafter{\romannumeral #1}}\def\u#1{{\accent"15
  #1}}\def\cprime{$'$} \def\cprime{$'$}
\begin{thebibliography}{10}

\bibitem{Abbott}
L.~F. Abbott.
\newblock Introduction to the background field method.
\newblock {\em Acta Phys. Polon. B\/} {\bf 13} (1982), 33--68.

\bibitem{Adler}
S.~L. Adler.
\newblock Short distance behavior of quantum electrodynamics and a eigenvalue
  condition for alpha.
\newblock {\em Phys.~Rev.~D\/} {\bf 5} (1972), 3021--3047.

\bibitem{bellonI}
M.~P. Bellon, G.~S. Lozano, and F.~A. Schaposnik.
\newblock Higher loop renormalization of a supersymmetric field theory.
\newblock {\em Phys.~Lett.~B\/} {\bf 650} (2007), 293--297.
\newblock ArXiv:hep-th/0703185.

\bibitem{bellonII}
M.~P. Bellon and F.~A. Schaposnik.
\newblock Higher loop renormalization of a supersymmetric field theory {II}.
\newblock ArXiv:0801.0727.

\bibitem{bergk}
C.~Bergbauer and D.~Kreimer.
\newblock {H}opf algebras in renormalization theory: {L}ocality and
  {D}yson-{S}chwinger equations from {H}ochschild cohomology.
\newblock {\em IRMA Lect. Math. Theor. Phys.\/} {\bf 10} (2006), 133--164.
\newblock ArXiv:hep-th/0506190.

\bibitem{bk30}
D.~Broadhurst and D.~Kreimer.
\newblock Combinatoric explosion of renormalization tamed by {H}opf algebra:
  30-loop {P}ade-{B}orel resummation.
\newblock {\em Phys. Lett. B\/} {\bf 475} (2000), 63--70.
\newblock ArXiv:hep-th/9912093.

\bibitem{bkerfc}
D.~Broadhurst and D.~Kreimer.
\newblock Exact solutions of {D}yson-{S}chwinger equations for iterated
  one-loop integrals and propagator-coupling duality.
\newblock {\em Nucl. Phys. B\/} {\bf 600} (2001), 403--422.
\newblock ArXiv:hep-th/0012146.

\bibitem{BDK}
D.~J. Broadhurst, R.~Delbourgo, and D.~Kreimer.
\newblock Unknotting the polarized vacuum of quenched {QED}.
\newblock {\em Phys. Lett. B\/} {\bf 366} (1996), 421--428.
\newblock ArXiv:hep-ph/9509296.

\bibitem{ckII}
A.~Connes and D.~Kreimer.
\newblock Renormalization in quantum field theory and the {R}iemann-{H}ilbert
  problem. {II}: The beta-function, diffeomorphisms and the renormalization
  group.
\newblock {\em Commun. Math. Phys.\/} {\bf 216} (2001), 215--241.
\newblock ArXiv:hep-th/0003188.

\bibitem{gkls}
S.~Gorishny, A.~Kataev, S.~Larin, and L.~Surguladze.
\newblock The analytic four-loop corrections to the {QED} $\beta$-function in
  the {MS} scheme and the {QED} $\psi$-function. {T}otal reevaluation.
\newblock {\em Phys. Lett. B\/} {\bf 256} (1991), 81--86.

\bibitem{ipz}
C.~Itzykson, G.~Parisi, and J.-B. Zuber.
\newblock Asymptotic estimates in quantum electrodynamics.
\newblock {\em Phys. Rev. D\/} {\bf 16} (1977), 996--1013.

\bibitem{jbw}
K.~Johnson, M.~Baker, and R.~Willey.
\newblock Self-energy of the electron.
\newblock {\em Phys.~Rev B\/} {\bf 136} (1964), B1111--B1119.

\bibitem{habil}
D.~Kreimer.
\newblock Renormalization and knot theory.
\newblock {\em J. Knot Theor. Ramifications\/} {\bf 6} (1997), 479--581.
\newblock ArXiv:q-alg/9607022.

\bibitem{hopf}
D.~Kreimer.
\newblock On the {H}opf algebra structure of perturbative quantum field
  theories.
\newblock {\em Adv. Theor. Math. Phys.\/} {\bf 2} (1998), 303--334.
\newblock ArXiv:q-alg/9707029.

\bibitem{anatomy}
D.~Kreimer.
\newblock Anatomy of a gauge theory.
\newblock {\em Annals Phys.\/} {\bf 321} (2006), 2757--2781.
\newblock ArXiv:hep-th/0509135v3.

\bibitem{tor}
D.~Kreimer.
\newblock {D}yson {S}chwinger equations: {F}rom {H}opf algebras to number
  theory.
\newblock {\em Fields Inst. Commun.\/} {\bf 50} (2007), 225--248.
\newblock ArXiv:hep-th/0609004.

\bibitem{linetude}
D.~Kreimer.
\newblock Etude for linear {D}yson-{S}chwinger equations.
\newblock In: S.~Albeverio, M.~Marcolli, S.~Paycha, and J.~Plazas, eds., {\em
  Traces in Geometry, Number Theory and Quantum Fields\/}, number E 38 in
  Aspects of Mathematics (Vieweg Verlag, 2008).

\bibitem{etude}
D.~Kreimer and K.~Yeats.
\newblock An \'etude in non-linear {D}yson-{S}chwinger equations.
\newblock {\em Nucl. Phys. B Proc. Suppl.\/} {\bf 160} (2006), 116--121.
\newblock ArXiv:hep-th/0605096.

\bibitem{radii}
D.~Kreimer and K.~Yeats.
\newblock Recursion and growth estimates in renormalizable quantum field
  theory.
\newblock {\em Commun. Math. Phys.\/} {\bf 279} (2008), 401--427.
\newblock ArXiv:hep-th/0612179.

\bibitem{MT}
G.~Mack and I.~T. Todorov.
\newblock Conformal-invariant green functions without ultraviolet divergences.
\newblock {\em Phys.Rev.D\/} {\bf 6} (1973), 1764--1787.

\bibitem{Nakanishi}
N.~Nakanishi and I.~Ojima.
\newblock {\em Covariant Operator Formalism of Gauge Theories and Quantum
  Gravity\/}, volume~27 of {\em Lect.\ Notes in Physics\/} (Singapore: World
  Scientific, 1990).

\bibitem{vBSW}
G.~van Baalen, A.~Schenkel, and P.~Wittwer.
\newblock Asymptotics of solutions in {$nA+nB\to C$} reaction-diffusion
  systems.
\newblock {\em Comm. Math. Phys.\/} {\bf 210} (2000), 145--176.

\bibitem{vS}
W.~D. van Suijlekom.
\newblock Renormalization of gauge fields: A {H}opf algebra approach.
\newblock {\em Commun. Math. Phys.\/} {\bf 276} (2007), 773--798.
\newblock ArXiv:hep-th/0610137.

\bibitem{vS2}
W.~D. van Suijlekom.
\newblock Renormalization of gauge fields using {H}opf algebras.
\newblock In: B.~Fauser, J.~Tolksdorf, and E.~Zeidler, eds., {\em Recent
  Developments in Quantum Field Theory\/} (Birkhauser Verlag, 2008).
\newblock ArXiv:0801.3170v1.

\bibitem{Weinberg}
S.~Weinberg.
\newblock {\em The Quantum Theory of Fields\/}, volume~II (Cambridge University
  Press, 1996).

\bibitem{kythesis}
K.~A. Yeats.
\newblock Growth estimates for {D}yson-{S}chwinger equations.
\newblock Ph.D. thesis, Boston University (2008).

\end{thebibliography}

\end{document}